\documentclass[conference]{IEEEtran}
\usepackage{graphics} 
\usepackage{epsfig} 
\usepackage{amsmath} 
\usepackage{amssymb}  
\usepackage{color}  

\newtheorem{lemma}{\textbf{Lemma}}
\newtheorem{theorem}{\textbf{Theorem}}
\newtheorem{remark}{\textbf{Remark}}

\usepackage{subfigure}

 \usepackage{comment}

\ifCLASSINFOpdf
\else
\fi

\begin{document}
\title{\LARGE Robotic Message Ferrying for Wireless Networks using Coarse-Grained Backpressure Control }

\author{
\IEEEauthorblockN{Shangxing Wang}
\IEEEauthorblockA{
Dept. of Electrical Engineering, \\
University of Southern California, \\
Los Angeles, CA\\
Email: shangxiw@usc.edu
}
\and
\IEEEauthorblockN{Andrea Gasparri}
\IEEEauthorblockA{Department of Engineering\\
Roma Tre University \\
Rome, Italy\\
Email: gasparri@dia.uniroma3.it}
\and
\IEEEauthorblockN{Bhaskar Krishnamachari}
\IEEEauthorblockA{
Dept. of Electrical Engineering, \\
University of Southern California, \\
Los Angeles, CA\\
Email: bkrishna@usc.edu}

}

\maketitle

\begin{abstract}
We formulate the problem of robots ferrying messages between statically-placed
source and sink pairs that they can communicate with wirelessly. We first analyze 
the capacity region for this problem under both ideal (arbitrarily high velocity, 
long scheduling periods) and realistic conditions. We indicate how robots could be 
scheduled optimally to satisfy any arrival rate in the capacity region,
given prior knowledge about arrival rates. We find that if the number of 
robots allocated grows proportionally with the number of source-sink pairs, then the capacity of 
the network scales as $\Theta(1)$, similar to what was shown previously by Grossglauser and
Tse for uncontrolled mobility; however, in contrast to that prior result, we also find that 
with controlled mobility this constant capacity scaling can be obtained while ensuring finite delay.
We then consider the setting where the arrival rates are unknown and present a coarse-grained backpressure message ferrying algorithm (CBMF) for it. In CBMF, the robots are matched to sources and sinks once every epoch to maximize a queue-differential-based weight. The matching controls both motion 
and transmission for each robot: if a robot is matched to a source, 
it moves towards that source and collects data from it; and if it is matched to a sink, it 
moves towards that sink and transmits data to it. We show through analysis
and simulations the conditions under which CBMF can stabilize the network. We show that 
the maximum achievable stable throughput with this policy tends to the ideal capacity 
as the schedule duration and robot velocity increase. 
\end{abstract}

\IEEEpeerreviewmaketitle

\section{Introduction}

Since the work by Tse and Grossglauser~\cite{tse02}, it has been known that the use of delay tolerant mobile communications
can dramatically increase the capacity of wireless networks by providing ideal constant throughput
scaling with network size at the expense of delay. However, nearly all the work to date has focused on 
message ferrying in intermittently connected mobile networks where the mobility is either unpredictable, or
predictable but uncontrollable. With the rapidly growing interest in multi-robot systems, we are entering an 
era where the position of network elements can be explicitly controlled in order to improve communication performance. 

This paper explores the fundamental limits of robotically controlled message ferrying in a wireless network. We consider
a setting in which a set of $K$ pairs of static wireless nodes act as sources and sinks that communicate not directly with each other 
(possibly because they are located far from each other and hence cannot communicate with each other at sufficiently high 
rates) but through a set of $N$ controllable robots. We assume that there is a centralized control plane (which, 
because it collects only queue state information about all network entities, can be relatively inexpensively created 
either using infrastructure such as cellular / WiFi, or through a low-rate multi-hopping mesh overlay).

We mathematically characterize the capacity region of this system, considering both ideal (arbitrarily large) and realistic (finite) settings with respect to robot mobility and scheduling durations. This analysis shows that with $N=2K$ robots the system could
be made to operate at full capacity (effectively at the same throughput as if all sources and sinks were adjacent to each other). 
We indicate how any traffic that is within the capacity region of this network can be served stably if the data arrival rates are known 
to the scheduler. We then consider how to schedule the robots when the arrival rates are not known \emph{a priori}. For this case, we 
propose and evaluate a queue-backpressure based algorithm for message ferrying that is coarse-grained in the sense that robot motion and relaying decisions are made once every fixed-duration epoch. We show that as the epoch duration and velocity of robots both increase, the throughput performance of this algorithm rapidly approaches that of the ideal case. 

\section{Problem Formulation}

There are $K$ pairs of static source and destination nodes located at arbitrary locations. Let the source for the $i^{th}$ flow be denoted as $src(i)$, and the destination or sink for that flow be denoted as $sink(i)$. Source $i$ receives packets at a constant rate denoted by $\lambda_i$. 

There are $N \le 2K$ mobile robotic nodes that act as message ferries, i.e. when they talk to a source node, they can collect packets from it, and when they talk to a sink node, they can transmit packets to it. Furthermore, for simplicity, we assume that the static nodes do not communicate directly with each other, but rather only through the mobile robots.  

\begin{figure}
    \centering
   \includegraphics[width=.45\textwidth]{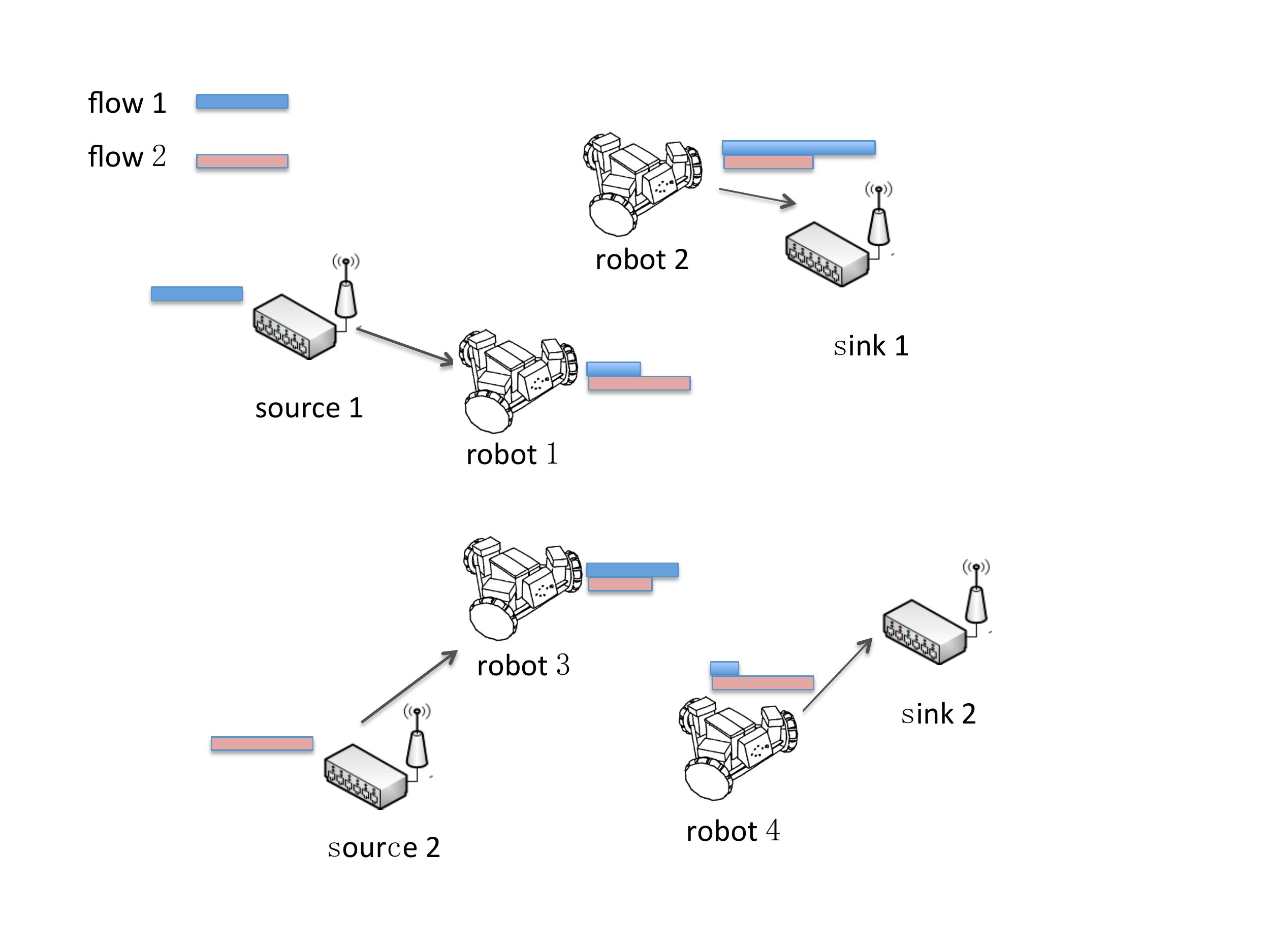}
    \caption{A network containing 2 pairs of source and sink nodes and 4 robots}
	\label{fig:networkIllustration}
\end{figure}

Time is divided into discrete time steps of unit duration. The locations of the sources and sinks for flow $i$ are denoted by $x_{src(i)}$ and $x_{sink(i)}$ respectively, and the location of robot $j$ at time $t$ is denoted as $x_j(t)$. Let the distance between a source for flow $i$ and a robot $j$ be denoted as $d(x_{src(i)}, x_j(t))$ (similarly for the sink). When in motion, the robotic nodes move with a uniform velocity $v$ directly to the destination (there are no obstacles), so that if robot $j$ is moving towards the source for flow $i$, its position $x_j(t)$ is updated so that it moves along the vector between its previous position and the source location to be at the following distance: 
\begin{equation} d(x_{src(i)}, x_j(t+1)) = max\{ d(x_{src(i)}, x_j(t)) - v , 0 \} \end{equation}

We assume that the rate at which a source for flow $i$ can transmit to a robot $j$, denoted by $R_{src(i), j}(t)$ is always strictly positive, and decreases monotonically with the distance between them, and similarly for the rate at which a robot $j$ can transmit to the sink for flow $i$, denoted by $R_{j, sink(i)}(t)$. We assume that when the robot is at a location of a particular source or sink, (i.e., the distance between them is 0), the corresponding throughput between the mobile robot and that source or sink is $R_{max}$

The queue at the source for flow $i$ is denoted as $Q_{src(i)}$. It is assumed that there is no queue at the sinks as they directly consume all packets intended for them. Each robot $j$ maintains a separate queue for each flow $i$, labelled $Q_j^i$. Figure~\ref{fig:networkIllustration} shows an illustration of this system with $K=2$ flows and $N=4$ robots. 

Every $T$ time steps there is a new epoch. At the start of each epoch, it is assumed that the information about queue states of all source and sink nodes as well as all queues at each of the robots is made available to a centralized scheduler. At that time this centralized scheduler can use this information to match each robot to either a source or sink. The matching is represented by an allocation matrix $A$ such that $A(i,j)$ is $0$ if the robot $j$ is not allocated to either source or sink for flow $i$,  $1$ if it is allocated to $src(i)$, and $-1$ if it is allocated to $sink(i)$. When a robot is allocated to a given source (or sink), for the rest of that epoch it moves closer to that node until it reaches its position. At all time steps of that epoch that robot will communicate exclusively with that source (or sink) to pick up (or drop, in case of the sink) any available packets between the corresponding queues at a rate depending on its current distance to that node. 

If a robot $j$ is communicating with $src(i)$ at time $t$, the update equations for the corresponding queue of the robot and the source queue will be as follows:
\begin{equation}
\begin{aligned}
n_p(t) & =  \min \{ R_{src(i), j}(t), Q_{src(i)}(t) \} \\
Q_j^i(t+1) & =  Q_j^i(t+1) + n_p(t) \\
Q_{src(i)}(t+1) & =  Q_{src(i)}(t+1) - n_p(t) + \lambda_i 
\end{aligned}.
\end{equation}

Similarly, if the robot $j$ is communicating with $sink(i)$ at time $t$, the queue update equation for the robot's corresponding queue will be: 
\begin{equation}
\begin{aligned}
n_p(t) & =  \min \{ R_{j, sink(i)}(t), Q_j^i(t) \} \\
Q_j^i(t+1) & =  Q_j^i(t+1) - n_p(t)
\end{aligned}.
\end{equation}

\section{Capacity Analysis}

We define an open region $\mathbf{\Lambda}$ of arrival rates as follows:
\begin{equation}
 \mathbf{\Lambda} = \left\{\mathbf{\lambda} | 0 \le \lambda_i < {R_{\max}}, \quad \forall~i,~\sum\limits_{i=1}^K \lambda_i < \frac{R_{\max}\,N}{2}\right\}.
\end{equation}

We shall show that this arrival rate region $\mathbf{\Lambda}$ can be served by a convex combination of configurations in which robots are allocated to serve distinct flows. Let $\tilde{\Gamma}$ be a finite set of vectors defined as:
\begin{equation}
\tilde{\Gamma} = \left\{ \gamma | \gamma_i = \frac{a_i \, R_{\max}}{2}, \quad \forall~i, a_i \in \{0, 1, 2\}, \sum\limits_{i=1}^K a_i \le N  \right\}. 
\end{equation}

For each element of this set $\tilde{\Gamma}$, the corresponding integer vector $\mathbf{a}$ corresponds to a ``basis" allocation of robots to distinct sources and sinks that can service each flow at rate $\gamma_i$. Specifically, $a_i$ refers to the number of robots allocated to serve flow $i$. If $a_i = 1$, this means exactly one robot is allocated to flow $i$, and can serve this flow maximally by spending half its time near the source and half the time near the sink (ignoring for now the time spent in transit), yielding a maximum service rate of $\gamma_i  = R_{max}/2$. If two robots are allocated  to a flow $i$, we have that $a_i = 2$, in which case two robots take turns spending time at the source and sink of the flow respectively for half the time each, yielding a net rate of $\gamma_i = R_{max}$. The constraints on $a_i$ ensure that the total number of robots allocated does not exceed the available number $N$. 

Let us refer to the convex hull of $\tilde{\Gamma}$ as $\mathcal{H}(\tilde{\Gamma})$ or, for readability, simply $\mathcal{H}$.

\begin{lemma}\label{lemma:H}
$\mathcal{H} \supset \Lambda$
\end{lemma}
\begin{proof} First, note that the convex hull of $\tilde{\Gamma}$ can be written as follows:
\begin{equation}
\mathcal{H} = \left\{ \gamma | \gamma_i = \frac{a_i \, R_{max}}{2}, a_i \in [0,2] \quad \forall~i, \sum\limits_{i=1}^K a_i \le N  \right\}.
\end{equation}

In other words, the convex hull of the set $\tilde{\Gamma}$ is obtained by allowing $a_i$ to vary continuously. Now using the relationship $a_i = \dfrac{2\gamma_i}{R_{max}}$, we can re-express $\mathcal{H}$ as follows:
\begin{equation}
\begin{aligned} 
\mathcal{H}  & = \left\{ \gamma |  \frac{2\gamma_i}{R_{max}} \in [0,2]~\forall~i, \sum\limits_{i=1}^K \frac{2\gamma_i}{R_{max}} \le N  \right\} \nonumber \\ 
 & = \left\{ \gamma |  0 \le \gamma_i \le R_{max}~\forall~i, \sum\limits_{i=1}^K \gamma_i \le \frac{R_{max} N}{2}   \right\}  \supset \Lambda  
\end{aligned}.
\end{equation} \end{proof}

Each basis allocation corresponding to the elements of $\tilde{\Gamma}$ can actually be expressed as two distinct but symmetric allocations of robots to sources/sinks over two successive epochs. For the $i^{th}$ flow, if $a_i = 0$, there is no robot allocated to either the source or sink in either of these two epochs; if $a_i = 1$, a particular robot is assigned to be at the source at the first epoch and at the sink at the second epoch; if $a_i = 2$, two robots are assigned (call them $R1$ and $R2$) such that $R1$ is at the source at the first epoch and at the sink at the second epoch while $R2$ is at the sink at the first epoch and at the source at the second epoch.

The set $\mathcal{H}$ describes all possible robot service rates that can be obtained by a convex combination of these basis allocations. Consider a rate vector $\gamma \in \mathcal{H}$. Since it lies in the convex hull of the set $\tilde{\Gamma}$ it can be described in terms of a vector of convex coefficients $\alpha$ each of whose elements corresponds to a basis allocation of robots. We can therefore\footnote{Here, for ease of exposition, we are assuming that $\alpha_i$ is rational, otherwise it can be approximated by an arbitrarily close rational number which will not affect the overall result.} identify $n_i$ such that ${n_i}/{\sum\limits_i n_i} = \alpha_i$. The given rate vector $\gamma$ can then be scheduled by allocating $n_i$ epochs each for the two parts of the $i^{th}$ basis allocation. And after a total of $\sum\limits_i 2 n_i$ epochs, the whole schedule can be repeated. This schedule will provide the desired service rate vector $\gamma$. 

Thus far the schedules have been derived under the assumption of instantaneous robot movements. Now we consider the effect of transit time. It is possible to choose $T$ or $v$ to be sufficiently large to bound the fraction of time spent in transit by $\epsilon$, i.e. $ \dfrac{d}{vT} < \epsilon$. Thus even while taking into account time wasted in transit, we can scale either time period of the epochs $T$ or the velocity $v$ so as to provide a service rate vector $\gamma'$ that is arbitrarily close to any ideal service rate $\gamma$ in the sense that $\gamma_i - \gamma'_i < \epsilon ~\forall~i$.

We now state one of our main results:

\begin{theorem}\label{thm:capacity}
$\mathbf{\Lambda}$ is the achievable capacity region of the network.
\end{theorem}

\begin{proof} By construction, $\mathcal{H}$ represents the boundary of all feasible robot service rates, and as we have discussed time spent in transit can be accounted for by increasing $T$ or $v$ so that any arrival rate that is in the interior of $\mathcal{H}$ can be served. Since in lemma~\ref{lemma:H}, we have already shown that $\Lambda \subset \mathcal{H}$, any arrival rate in $\Lambda$ can be stably served. 

Furthermore, $\mathcal{H}$ represents the closure of the open set $\Lambda$. Thus any arrival rate vector that is a bounded distance outside of $\Lambda$ cannot be served stably (as it would also be outside of $\mathcal{H}$). 

Together, these imply that $\Lambda$ is the achieveable capacity region of the network. \end{proof}

\subsection{An Example}

Figure~\ref{fig:capacityIllustration} shows the capacity region when $K = 2$, $N = 3$.  The labels such as $(x, y)$ are given to the basis allocations on the Pareto boundary to denote that they can be achieved by allocating an integer number of robots $x$ to flow $1$ and $y$ to flow $2$. Note in particular that the point $(R_{max}, R_{max})$ is outside the region in this case because the only way to serve that rate is to allocate two robots full time to each of the two flows, and we have only 3 robots. The vertices on the boundary of the region, which represent basis allocations, are all in the set $\tilde{\Gamma}$; the convex hull of $\tilde{\Gamma}$ completely describes the region.

\begin{figure}
    \centering
	\includegraphics[width=.35\textwidth]{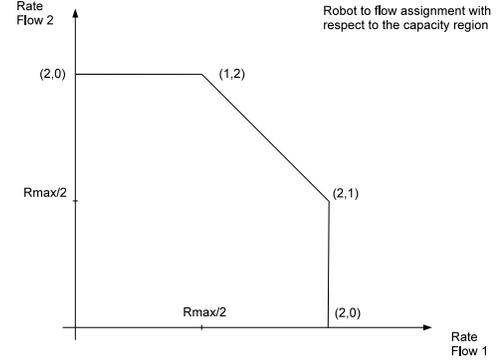}
	\caption{Capacity region for a problem with 3 robots and 2 flows}
	\label{fig:capacityIllustration}
\end{figure}

\subsection{$\Theta(1)$ Capacity Scaling with Controlled Mobility}

In~\cite{tse02}, Grossglauser and Tse first showed that in a network with uncontrolled mobility, under certain mixing conditions, a total capacity of $\Theta(1)$ could be achieved by using one intermediate relay node. Our modeling in this paper shows that the total capacity region scales linearly with the number of robotic relays. Therefore, when the number of robots is linear in the number of flows, the per-pair network capacity will be $\Theta(1)$ here as well. 

There are a few minor differences between the model in this paper and what is considered in~\cite{tse02}, however these differences are not consequential, as they affect only constants in the asymptotic scaling:

\begin{itemize}

\item In~\cite{tse02} it is assumed that all nodes are mobile. In our setting, first note that if the source and sink nodes were \emph{controllably} mobile, then they could each be simply paired up directly and moved arbitrarily close to each other, and we would achieve $\Theta(1)$ scaling without even needing the controllable relays. Even if the source and sink were randomly moving, if they do so within a bounded region in such a way that the controllable robots could always locate and move to them within a finite time, our results would be remain unaffected. 

\item In~\cite{tse02} it is assumed that each source/sink node is a source for one flow and a sink for another. Making the same assumption in our model for the static nodes would merely double the number of flows, and would result only in a constant factor difference.

\item For ease of exposition and analysis, in our work we have assumed that robots do not interfere with each other at any time. However, for deriving the capacity region, it suffices to assume that the robots do not interfere with each other whenever they are arbitrarily close to the source/sink they are communicating with. This is consonant with the modeling and result in~\cite{tse02} that when nodes are sufficiently close to each other they may communicate without experiencing interference from any number of other distant transmitters. 
\end{itemize}


Finally, in~\cite{tse02}, the $\Theta(1)$ capacity is obtained at the cost of average delay increasing with the size of the network. In stark contrast, in our formulation, as we discuss below, it is still possible to obtain a constant fraction of the full capacity (hence still maintaining the $\Theta(1)$ capacity scaling) even while keeping the delay bounded.

\subsection{Capacity Region under finite velocity and epoch duration}

The analysis thus far assumes that either the velocity of the robot or the epoch duration can be chosen to be arbitrarily large. Next, motivated by practical considerations we consider the case when $v$ and $T$ are finite. In particular, the restriction of $T$ to be finite is useful for two reasons: a) it fixes the overhead of scheduling and b) it can be used to enforce a deterministic upper bound on delay (the time between generation and delivery of a given packet). As may be expected, these constraints reduce the capacity region. 

The fraction of time spent in transit, is bounded by $\dfrac{d}{vT}$, where $d$ is the maximum distance between the static nodes. We assume that $\dfrac{d}{vT} < 1$, which implies that a robot can always reach its destination (source or sink) within an epoch.  


This directly yields the following inner-bound on the capacity region when $v$ and $T$ are finite and fixed:
\begin{equation}
\begin{aligned}
 \mathbf{\Lambda_{IB}(v,T)} = \bigg \{ & \mathbf{\lambda} | 0 \le \lambda_i < R_{max}(1-\frac{d}{vT}),~\\
 &\forall~i,~\sum\limits_{i=1}^K \lambda_i < \frac{R_{max}(1-\frac{d}{vT})\,N}{2}\bigg\}
 \end{aligned}.
\end{equation}

\begin{remark}
 Any arrival rate in the inner-bound region can still be achieved while scheduling them this way. As the inner bound is only a constant factor away from the full capacity region, this shows that a capacity scaling of $\Theta(1)$ can be achieved with controllable mobility even while keeping average delay to be bounded. This is in contrast to what happens with opportunistic mobility~\cite{tse02} where a constant capacity scaling is obtained at the cost of unbounded delay. Note further that when the number of robots $N = 2K$, it is possible to schedule the robots for each flow in alternate cycles so that even the worst case delay is bounded deterministically by $2T$.
\end{remark}

\section{Coarse-Grained Backpressure Control}

From the previous discussion, we know that if the arrival rate of each flow is known, and within the ideal capacity region of the system, the epoch duration and a service schedule for the robots can be designed in such a way that the rate is served in a stable manner (maintaining the average size of each queue to be bounded). We consider now the case when the arrival rates are within the capacity region  but not known to the scheduling algorithm, and the $v$ and $T$ parameters are kept fixed. Is it still possible to schedule the movement and communications of the robots in such a way that all queues remain stable?

The answer to this question turns out to be yes, using the notion of Backpressure scheduling first proposed by Tassiulas and Ephremides~\cite{tassiulas92}. We propose an algorithm for scheduling message ferrying robots that achieves throughput-optimal performance for finite $v$ and $T$ parameters, which we refer to as coarse-grained backpressure-based message ferrying (CBMF). The CBMF algorithm works as follows. 

At the beginning of each epoch:
\begin{itemize}
\item compute the weights \mbox{$w_{src(i), j} = (Q_{src(i)} - Q^i_j)$} and \mbox{$w_{sink(i), j} = (Q^i_j)$}. 
\item If the allocation $A(i,j) = 1$, denote \mbox{$w_{i,j} = w_{src(i), j}$}. If \mbox{$A(i,j) = -1$}, denote \mbox{$w_{i,j} = w_{sink(i), j}$}. Else, if \mbox{$A(i,j) = 0$, $w_{i,j} = 0$}.
\item Find the allocation $A$ that maximizes $\sum_{i,j} |A(i,j)| w_{i,j}(A(i,j))$ subject to the following three constraints: 
\begin{itemize}
\item[(1)] $\sum_i |A(i,j)| = 1$, 
\vspace{2mm}
\item[(2)] $\sum_j \mathcal{I}\{A(i,j) = 1\} \le 1$,
\vspace{2mm}
\item[(3)] $\sum_j \mathcal{I}\{A(i,j) = -1\} \le 1$. 
\end{itemize}
The first constraint ensures that each robot is allocated to exactly one source or sink. The second constraint ($\mathcal{I\{\}}$ represents the indicator function) ensures that no source is allocated more than one robot, while the third constraint ensures that no sink is allocated more than one robot.
\end{itemize}

\begin{theorem}
For any arrival rate that is strictly within $\mathbf{\Lambda_{IB}(v,T)}$, the CBMF algorithm ensures that all source and robot queues are stable (always bounded by a finite value). 
\end{theorem}

\begin{proof}
The proof essentially follows the treatment in~\cite{NeelyBook}.

Since the arrival rate is strictly interior in $\mathbf{\Lambda_{IB}(v,T)}$, we can make some simple assumptions. We ignore data transmitted when robots are moving. And at the beginning of each epoch, once the robots are allocated, they move instantly to their destinations(sources or sinks) and remain static with  a constant transmission rate as $R_{max}' = R_{max}(1-\frac{d}{vT})$.

Let $b_{ij}(t)$ denote whether robot $j$ is allocated to $src(i)$ or not. $b_{ij}(t)=1$ means robot $j$ is allocated to $src(i)$; $b_{ij}(t)=0$ means robot $j$ is not allocated to $src(i)$. Similarly, $c_{ij}(t)$ denotes whether robot {j} is allocated to $sink(i)$ or not. $c_{ij}(t)=1$ means robot $j$ is allocated to $sink(i)$, and $c_{ij}(t)=0$ means robot $j$ is not allocated to $sink(i)$. Then we have $R_{src(i),j}(t) = b_{ij}(t)R_{max}'$ and $R_{j,sink(i)}(t)=c_{ij}(t)R_{max}'$.

The queue backlog at source $i$, $\forall i \in  \{1,...,K \}$,  is updated as follows
\begin{small}
\begin{equation}
Q_{src(i)}(t+1) \!\!= \!\!max\! \left \{Q_{src(i)}(t)-(b_{i1}(t)+...+b_{iN}(t))R_{max}',0 \right \} + \lambda_{i}.
\end{equation}
\end{small}
The queue backlog at robot $j$ for flow $i$, $\forall i \in {1,...,K} $ and $j \in {1,...,N}$, is given by
\begin{scriptsize}
\begin{equation}
Q_j^i (t+1) \!\!= \!\!  max \! \left \{Q_j^i(t)-c_{ij}(t)R_{max}',0 \right \} + \!min\! \left \{Q_{src(i)}(t), b_{ij}(t)R_{max}' \right \}.
\end{equation}
\end{scriptsize}

Define the queue backlog vector of this system as 
\begin{scriptsize}
\begin{equation}
\mathbf{\Theta}(t) \!\!=\!\! ( \! Q_{src(1)}(t),...,Q_{src(K)}(t), Q_1^1(t),...,Q_1^K(t), ..., Q_N^1(t), ..., Q_N^K(t) \!).
\end{equation}
\end{scriptsize}
\noindent And the Lyapunov function as 
\begin{equation}
L(\mathbf{\Theta}(t))=\frac{1}{2} \left [\sum\limits_{i=1}^K Q_{src(i)}(t)^2 + \sum\limits_{i=1}^{K} \sum\limits_{j=1}^{N} Q_j^i(t)^2 \right ].
\end{equation}

Then ,

\begin{eqnarray}
L(\mathbf{\Theta}(t+1))\!\!\!\!\! \!&-&\!\!\!\!\! \! L(\mathbf{\Theta}(t))  \nonumber \\
&=& \!\!\!\!\! \!\frac{1}{2} \!\! \left \{\sum\limits_{i=1}^{K} \left [Q_{src(i)}(t+1)^2 - Q_{src(i)}(t)^2 \right ]  \right \}  \nonumber \\
&+& \!\!\!\!\! \! \frac{1}{ 2} \left \{ \sum\limits_{i=1}^{K} \sum\limits_{j=1}^{N} \left [ Q_j^i(t+1)^2-Q_j^i(t)^2 \right ]  \!\! \right \}  \nonumber \\ 
&\leqslant&\!\!\!\!\! \! \sum\limits_{i=1}^{K} \frac{(\sum\limits_{j=1}^{N}b_{ij}(t)R_{max}')^2+\lambda_{i}^2}{2} \nonumber \\
&+& \!\!\!\!\! \!  \sum\limits_{i=1}^{K}\sum\limits_{j=1}^{N}\frac{(c_{ij}(t)R_{max}')^2 + (b_{ij}(t)R_{max}')^2}{2} \nonumber \\
&+&\!\!\!\!\! \! \sum\limits_{i=1}^{K} Q_{src(i)}(t) \left [\lambda_{i}-\sum\limits_{j=1}^{N}b_{ij}(t)R_{max}') \right ] \nonumber \\
&+& \!\!\!\!\! \!  \sum\limits_{i=1}^{K} \sum\limits_{j=1}^{N}Q_j^i(b_{ij}(t)-c_{ij}(t))R_{max}' .\!\!\!\!\! \!
\end{eqnarray}
\noindent where the inequality comes from equations (8) and (9), and 

\begin{equation}
max \left \{ Q-b,0  \right \} + a \leqslant Q^2 + a^2 + b^2 + 2Q(a-b).
\end{equation}
\begin{equation}
max \left \{ Q_1-c,0 \right \} + min \left \{ Q_2,b\right \} \leqslant max \left \{ Q_1-c,0 \right \} + b.
\end{equation}

Define the conditional Lyapunov drift as 
\begin{equation}
\bigtriangleup(\mathbf{\Theta}(t))= \mathbb{E} \left \{L(\mathbf{\Theta}(t+1)) - L(\mathbf{\Theta}(t)) | \mathbf{\Theta}(t) \right \}.
\end{equation}

Since $\forall i \!\! \in \!\!  \{1,...,K\}$ and $\forall j \! \! \in \!\! \{1,...,N\}$, $\qquad \qquad \qquad \qquad b_{ij}(t)R_{max}' \!\! \leqslant  \!\! R_{max}$, $c_{ij}(t)R_{max}' \!\! \leqslant \!\! R_{max}'$ and the arrival rates are finite, the first two terms on the left-hand-side of inequality (15) can be upper bounded by a finite constant $B$. Thus,
\begin{small}
\begin{eqnarray}
\!\!\!\!\!\!\! \!&\bigtriangleup &\!\!\!\!\!\!\! \! (\mathbf{\Theta}(t))\!\!\!\!\! \!\nonumber \\
&\leqslant&\!\!\!\!\! \! B + \sum\limits_{i=1}^{K} Q_{src(i)}(t)\lambda_i \nonumber \\
&-& \!\!\!\!\!\!\!\! \! \sum\limits_{i=1}^{K}\!\!\! \sum\limits_{j=1}^{N} \!\!  \mathbb{E} \!\! \left \{ (Q_{src(i)}(t)\!-\!Q_j^i(t))b_{ij}(t) R_{max}' \!+\! Q_j^i(t)c_{ij}(t) R_{max}' | \mathbf{\Theta}(t) \! \right \} . \nonumber \\
\end{eqnarray}
\end{small}

Applying the CBMF algorithm to allocate robots, the last term on the right-hand-side can be maximized, thus the conditional drift can be minimized. Let $b_{ij}^*(t)$ and $c_{ij}^*(t)$ be any other robot allocation policy, then we have 

\begin{small}
\begin{eqnarray}
\bigtriangleup \!\!\!\!\!\!\! \!&(&\!\!\!\!\!\!\! \! \mathbf{\Theta}(t)) \nonumber \\  
&\leqslant& \!\!\!\!\! \! B + \sum\limits_{i=1}^{K} Q_{src(i)}(t)\lambda_i \nonumber \\
&-& \!\!\!\!\! \! \! \sum\limits_{i=1}^{K} \! \sum\limits_{j=1}^{N} \mathbb{E}  \!\! \left \{ (Q_{src(i)}(t)\!\!-\!\!Q_j^i(t))b_{ij}^*(t) R_{max}' \!\! +\!\!  Q_j^i(t)c_{ij}^*(t) R_{max}' | \mathbf{\Theta}(t) \right \} \nonumber \\
 \!\!\!\!\! \!&\leqslant& \!\!\!\!\! \! B  \!\! -  \!\!\sum\limits_{i=1}^K \!Q_{src(i)}(t) \!\!  \left (  \mathbb{E}\! \left \{  \sum\limits_{j=1}^{N} b_{ij}^*(t)R_{max}' | \mathbf{\Theta}(t) \right \} - \lambda_i \!\! \right) \nonumber \\
  \!\!\! &-& \!\!\! \sum\limits_{i=1}^{K} \sum\limits_{j=1}^{N} Q_j^i(t)\mathbb{E} \left \{  \left (  c_{ij}^*(t)-b_{ij}^*(t)  \right ) R_{max}' | \mathbf{\Theta}(t)    \right \}. \nonumber \\
\end{eqnarray} 
\end{small}

In order to upper bound the terms on the right-hand-side, let us first consider the following problem: given an arrival rate vector $\mathbf{\lambda} = (\lambda_1, ..., \lambda_K) \in \mathbf{\Lambda_{IB}(v,T)}$ \emph{a priori}, we want to design an S-only (depends only on the channel states) algorithm to

\begin{eqnarray}
\text{find} &&\quad \quad \qquad  \epsilon  >   0 \nonumber \\
\text{s.t.} && \!\!\!\!\!\!\!\!\!\!\! \! \!\!\!\!\! \! \!\!\!\!\! \!  \quad \quad \qquad \lambda_i + \epsilon \leqslant \mathbb{E} \left \{  \sum\limits_{j=1}^{N} b_{ij}^*(t)R_{max}' \right \}  \quad \forall i \in \{1,...,K\}, \nonumber \\
&& \mathbb{E} \left \{ b_{ij}^*(t)R_{max}' \right \} + \epsilon \leqslant \mathbb{E} \left \{ c_{ij}^*(t)R_{max}' \right \}, \nonumber \\
 && \forall i \in \{1,...,K\} \quad \text{and} \quad \forall j \in \{1,...,N \}.\quad 
\end{eqnarray}


The S-only algorithm to achieve any given arrival rates which are strictly interior to the capacity region is designed as follows:

Since $\mathbf{\lambda} = (\lambda_1, ..., \lambda_K) \in \mathcal{H} / \partial \mathcal{H} $, we can find a vector $\mathbf{\epsilon} = (\epsilon_1, ..., \epsilon_K)$ such that $\mathbf{\lambda}' = (\lambda_1+\epsilon_1, ..., \lambda_K + \epsilon_K) \in \partial \mathcal{H}$. Let $\epsilon_{max} = min \{\epsilon_1, ..., \epsilon_K\}$, and since $\mathbf{\lambda}$ is strictly interior in $\mathcal{H}$, we have $\epsilon_{max} > 0$.

Since $\lambda'' = (\lambda_1+\epsilon_{max}, ..., \lambda_K+\epsilon_{max}) \in \mathcal{H}$, it can be represented as a convex combination of basis allocations. To be specific, in a network containing $K$ flows and $N$ robots, there are $M$ (depends on $K$ and $N$, and is finite) basis allocations in total. Let $[\lambda_{l1},..., \lambda_{lK}]^T$, $\forall l \in \{1,...,M \}$ denote the capacity the $l_{th}$ allocation can provide.  Let $\mathbf{\alpha} = (\alpha_1, ..., \alpha_M)$ be the allocation vector of the convex coefficients. And we have 
\begin{eqnarray}
\alpha_1 \!\!\!\!\!\!\! &[&\!\!\!\!\!\!\! \lambda_{11},..., \lambda_{1K}]^T + ... + \alpha_M [\lambda_{M1},..., \lambda_{MK}]^T \nonumber \\
\!\!\!\!\!\!&=&\!\!\!\!\!\! [\lambda_1 + \epsilon_{max}, ..., \lambda_K+\epsilon_{max}]^T.
\end{eqnarray}

Let us identify integers $n_l$ satisfying $\frac{n_l}{\sum\limits_{i=1}^{M} n_i} = \alpha_l$, $\forall l \in \{1,...,M \}$. Then the arrival rate vector $\mathbf{\lambda}''$ can be served by first allocating $n_l$ epochs for the $l_{th}$ basis allocation, and allocating the next $n_l$ epochs for the same $l_{th}$ basis allocation but exchanging the robots locations, $\forall l \in \{1,...,M \}$. And after a total of $2\sum\limits_{l=1}^{M} n_l$ epochs, repeat the whole process. 

Since the served rate $\lambda''$ is $\epsilon_{max}$ greater than the original given rate $\lambda$, we can change the above scheduling scheme by adding a few more epochs to each $2\sum\limits_{l}n_{l}$ epochs period, which still supports the given input rate. During these additional epochs, we can evenly allocate one robot at each sink to help deliver data. In this way we can make sure there exists an $\epsilon' > 0$ such that

\begin{small}
\begin{eqnarray}
&& \lambda_i + \epsilon'  \leqslant    \mathbb{E} \left \{  \sum\limits_{j=1}^{N} b_{ij}^*(t)R_{max}' \right \}, \, \forall i \in \{1,...,K\},  \nonumber \\
&& \mathbb{E} \left \{ b_{ij}^*(t)R_{max}' \right \} + \epsilon'  \leqslant  \mathbb{E} \left \{ c_{ij}^*(t)R_{max}' \right \}, \nonumber  \\
&&  \forall i \in \{1,...,K\} \, \text{and} \, \forall j \in \{1,...,N \}.
\end{eqnarray}
\end{small}

Take inequalities (20) into equation (17), we have
\begin{equation}
\bigtriangleup (\mathbf{\Theta}(t)) \leqslant B - \epsilon'  \left [ \sum\limits_{i=1}^{K} Q_{src(i)} (t)+ \sum\limits_{i=1}^{K} \sum\limits_{j=1}^{N}Q_j^i(t) \right ].
\end{equation}

Taking iterated expectations, summing the telescoping series, and rearranging terms yields:

\begin{eqnarray}
\frac{1}{t} \sum\limits_{\tau=0}^{t-1}\!\!\!\!\!\!\! \!&&\!\!\!\!\!\!\! \! \left [  \sum\limits_{i=1}^{K} \mathbb{E} \left \{ Q_{src(i)}(\tau) \right \}  + \sum\limits_{i=1}^{K}\sum\limits_{j=1}^{N} \mathbb{E} \left \{ Q_j^i(\tau) \right \} \right] \nonumber \\
\!\!\!\!&\leqslant&\!\!\!\! \frac{B}{\epsilon'} + \frac{\mathbb{E}\{L(\mathbf{\Theta}(0)) \} }{\epsilon't} .
\end{eqnarray}

Therefore, 
\begin{equation}
\lim_{t \rightarrow \infty}\frac{1}{t} \sum\limits_{\tau=0}^{t-1} \left [  \sum\limits_{i=1}^{K} \mathbb{E} \left \{ Q_{src(i)}(\tau) \right \} + \sum\limits_{i=1}^{K}\sum\limits_{j=1}^{N} \mathbb{E} \left \{ Q_j^i(\tau) \right \}  \right] \leqslant \frac{B}{\epsilon'} ,
\end{equation}
\noindent which indicates that the system is strongly stable. 
\end{proof}

\begin{remark}
CBMF is  provably stable for all arrival rates up to the inner-bound. However, as $v$ and $T$ increase, the inner-bound approaches the ideal capacity region.
\end{remark}



\section{Delay Analysis for Single-Flow Two-Robots}

Consider a system which contains one source $S$, one sink $D$, and two robots $R_1$, $R_2$. Initially, their queue backlogs are empty. The distance between source and sink is $d$. $R_1$ is at the same location as $D$, and $R_2$ is at the same location as $S$. The moving speeds of robots is $v$. The transmission rate function is $R( x )$. 

Without loss of generality, we assume that after applying CFBP at the beginning of the first epoch, $R_1$ (called receiving robot) is allocated to $S$ to receive data, and $R_2$ (called delivering robot) is allocated to $D$ to deliver data. Thus, in the first epoch, $R_1$ moves towards the source and keep receiving data. Since $R_2$ doesn't have any data stored in its queue, it just moves to the destination without delivering data. At the end of this epoch, $Q_s = 0$, $Q_{R_1} = \lambda T$ and $Q_{R_2} = 0$. 

In the second epoch, after applying CFBP, $R_1$ will be allocated to the destination to deliver data, and $R_2$ will be allocated to the source to receive data. Applying CFBP algorithm at the begging of each epoch, this whole process  repeats every two epochs in the flowing epochs. 

Since we focus on the stable state of the system in the long run, we can ignore the first epoch. What's more, the system evolves all the same in the following epochs except for changing the roles of $R_1$ and $R_2$(In one epoch, one is the receiving robot and the other is the delivering robot; In the next epoch, the other way around.). Thus, we can analyze one particular epoch instead. Let us focus on the second epoch. At the beginning of this epoch, queues at source $S$ and robot $R_2$ are both empty, and queue at $R_1$ is $\lambda T$. According to the CFBP algorithm, in the second epoch, robot $R_1$ is the delivering robot moving towards the sink and deliver data; and robot $R_2$ is the receiving robot moving towards the source to receive data. Depending on different values of arrival rate $\lambda$, the system evolves as one the following two cases.

Case 1: when arrival rate $\lambda$ is small enough so that the delivering robot $R_1$ can deplete all its packets before reaching the sink. Let's $t_1^*$ be the time when the queue at $R_1$ is empty. Then, $\lambda T - \int_0^{t_1^*} R(d-vt)dt = 0$, and $t_1^* \leqslant \frac{d}{v}$. Note the condition $t_1^* \leqslant \frac{d}{v}$ also provides an upper bound $\hat{\lambda}_{max}$ for the input rate in this case. 


1) When $0 \leqslant t_1 \leqslant t_1^*$, the delivering robot $R_1$ has data in its backlog and keeps transmitting data to $D$. Thus, its queue at time $t_1$ is $Q_{R_1}(t_1) = \lambda T - \int_{0}^{t_1} R(d-vt) dt$. At the same time, new data arrives at the source node $S$ at rate $\lambda$, and the receiving robot $R_2$ keeps receiving data from $S$. The total amount of data at $S$ and $R_2$ at $t_1$ is $Q_s(t_1) + Q_{R_2}(t_1) = \lambda t_1$. Thus, the total queue backlog in the system at $t_1$ is 
\begin{small}
\begin{equation}
Q_{tot}(t_1) = Q_s(t_1) + Q_{R_1}(t_1) +Q_{R_2}(t_1) = \lambda T - \int_{0}^{t_1} \! R(d-vt)dt + \lambda t_1. 
\end{equation}
\end{small}

2) When $t_1^*<t_1\leqslant T $, the delivering robot $R_1$ has no data to transmit, thus $Q_{R_1}(t_1) = 0$. The total amount of data at $S$ and $R_2$ at $t_1$ is still $Q_s(t_1) + Q_{R_2}(t_1) = \lambda t_1$. Thus, the total queue backlog in the system at $t_1$ is 
\begin{equation}
Q_{tot}(t_1) = \lambda t_1.
\end{equation}

By definition, the time average total queue is $\qquad \qquad \qquad \bar{Q}_{tot} = \frac{1}{T} \int_0^T Q(\tau) d \tau$. Then, we have
\begin{eqnarray}
\bar{Q}_{tot}  \!\!\!\!&=&\!\!\!\! \frac{1}{T} \left \{ \int_{0}^{t_1^*} \!\! \left \{  \lambda T \!\!- \!\! \int_{0}^{\tau} \! R(d-vt)dt \!\!+\!\! \lambda \tau \right \} d\tau + \int_{t_1^*}^{T} \lambda \tau d \tau \!\!\right \} \nonumber \\
\!\!\!\!&=&\!\!\!\! \lambda t_1^* + \frac{\lambda T}{2} - \frac{1}{T} \int_0^{t_1^*} \left \{ \int_0^{\tau} R(d-vt)dt \right \} d \tau.
\end{eqnarray}

And by Little's Law, the time average delay of this system is
\begin{equation}
\bar{D} = \frac{\bar{Q}_{tot}}{\lambda} =   t_1^* + \frac{T}{2} - \frac{1}{\lambda T} \int_0^{t_1^*} \left \{ \int_0^{\tau} R(d-vt)dt \right \} d \tau. 
\end{equation}

Case 2: when arrival rate $\lambda$ ($\hat{\lambda}_{max}<\lambda < \lambda_{max}$) is large enough so that the delivering robot $R_1$ cannot finish depleting all its data during moving, i.e., it will keep downloading data at rate $R_{max}$ while it reaches the destination. Let $t_2^*$ be the time when the queue backlog at $R_1$ is empty. Then, $\quad \quad \quad \quad \quad \quad \lambda T - \left \{ \int_{0}^{\frac{d}{v}} R(d-vt)dt + (t_2^* - \frac{d}{v})R_{max} \right \} = 0$. Note $t_2^* \leqslant T$ also provides the upper bound capacity $\lambda_{max}$ for the stable system. (where $ \lambda_{max} = R_{avg} = \frac{\int_{0}^{\frac{d}{v}}R(d-vt)dt + (T-\frac{d}{v})R_{max}}{T}$)

1) When $0 \leqslant t_1 \leqslant \frac{d}{v}$, the delivering robot $R_1$ keeps moving towards $D$ while transmitting data. Its queue backlog at $t_1$ is $Q_{R_1}(t_1) = \lambda T - \int_{0}^{t_1}R(d-vt)dt$. The queue at source $S$ and receiving robot $R_2$ at time $t_1$ is $Q_s(t_1) + Q_{R_2}(t_1) = \lambda t_1$. Thus, the total queue backlog in the system at $t_1$ is 
\begin{eqnarray}
Q_{tot}(t_1) \!\!\!\!&=&\!\!\!\! Q_s(t_1) + Q_{R_1}(t_1) +Q_{R_2}(t_1) \nonumber \\
\!\!\!\!&=&\!\!\!\! \lambda T - \int_{0}^{t_1} \! R(d-vt)dt + \lambda t_1. 
\end{eqnarray}

2) When $\frac{d}{v}<t_1 \leqslant t_2^*$, delivering robot $R_1$ stays at the same location with sink $D$ and keeps delivering data at the maximum rate $R_{max}$. Then its queue at $t_1$ is $R_1(t_1) = \lambda T - \int_{0}^{\frac{d}{v}}R(d-vt)dt -(t_1-\frac{d}{v})R_{max}$. The queue at $S$ and $R_2$ is $Q_s(t_1) + Q_{R_2}(t_1) = \lambda t_1$. Thus, the total queue backlog at $t_1$ is 
\begin{equation}
Q_{tot}(t_1) = \lambda T - \int_{0}^{\frac{d}{v}} \! R(d-vt)dt - (t_1 - \frac{d}{v})R_{max} + \lambda t_1.
\end{equation}

3) When $t_2^*<t_1 \leqslant T$, the delivering robot $R_1$ has already delivered all its data, thus $Q_{R_1}(t_1) = 0 $. Besides, we still have $Q_s(t_1) + Q_{R_2}(t_1) = \lambda t_1$. Thus, the total queue backlog is 
\begin{equation}
Q_{tot}(t_1) = \lambda t_1.
\end{equation}

Thus,
\begin{small}
\begin{eqnarray}
\bar{Q}_{tot} \!\!\!\! &=& \!\!\!\!  \frac{1}{T}  \{ \int_0^{\frac{d}{v}} \left \{   \lambda T - \int_{0}^{\tau}R(d-vt)dt   + \lambda \tau \right \} d \tau \nonumber \\
\!\!\!\!&+& \!\!\!\! \int_{\frac{d}{v}}^{t_2^*} \left \{ \lambda T    - \int_{0}^{\frac{d}{v}}R(d-vt)dt  - (\tau - \frac{d}{v})R_{max}  + \lambda \tau \right \} d \tau \nonumber \\
\!\!\!\! &+& \!\!\!\!  \int_{t_2^*}^{T} \lambda \tau d \tau  \} \nonumber \\
\!\!\!\! &=& \!\!\!\! \lambda t_2^* + \frac{\lambda T}{2} - \frac{1}{T} \int_0^{\frac{d}{v}} \left \{ \int_0^{\tau} R(d-vt)dt \right \} d \tau  \nonumber \\
\!\!\!\!&-&\!\!\!\! \frac{1}{T}(t_2^*-\frac{d}{v}) \int_{0}^{\frac{d}{v}} R(d-vt)dt - \frac{R_{max}}{2T}(t_2^*-\frac{d}{v})^2. \nonumber \\
\end{eqnarray}
\end{small}

and

\begin{eqnarray}
\bar{D} = \frac{\bar{Q}_{tot}}{\lambda}  \!\!\!\!&=&\!\!\!\!   t_2^* + \frac{T}{2} - \frac{1}{\lambda T} \int_0^{\frac{d}{v}} \left \{ \int_0^{\tau} R(d-vt)dt \right \} d \tau  \nonumber \\
\!\!\!\!&-&\!\!\!\! \frac{1}{\lambda T} (t_2^*-\frac{d}{v}) \int_{0}^{\frac{d}{v}} R(d-vt)dt - \frac{R_{max}}{2 \lambda T }(t_2^*-\frac{d}{v})^2. \nonumber \\
\end{eqnarray}

\begin{figure*}[tbh]
   \centering
       \vspace{-.75in}
   \includegraphics[width=.45\textwidth]{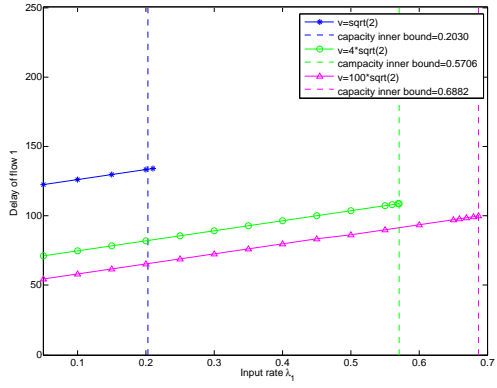}
\includegraphics[width=.45\textwidth]{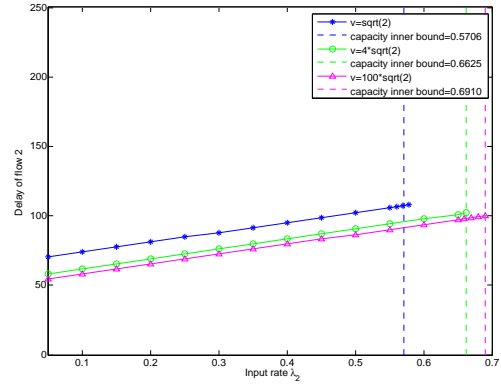}
    \vspace{-27.8mm}
     \caption{Delay of flows 1 (left) and 2 (right) as we vary $v$ for $T=100$}
    \label{fig:dvv}
\end{figure*}

\begin{figure*}[thb]
    \centering
    \vspace{-1.2in}
   \includegraphics[width=.45\textwidth]{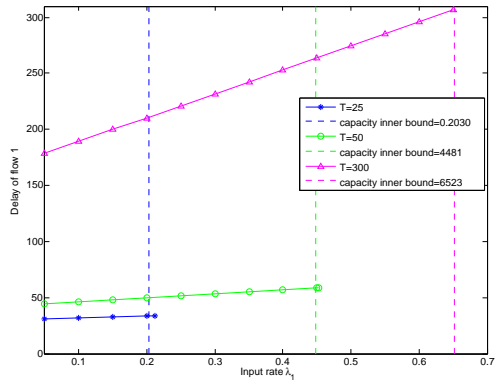}
  \includegraphics[width=.45\textwidth]{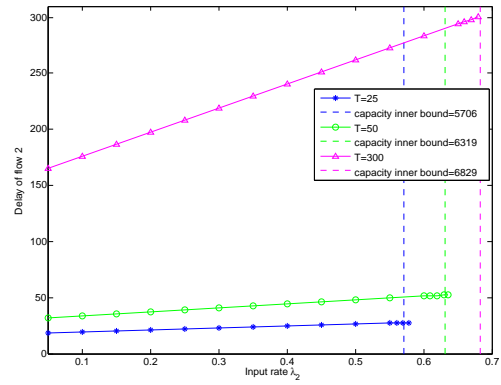}
    \vspace{-27.8mm}
 \caption{Delay of flows 1 (left) and 2 (right) as we vary $T$ for $v=4*\sqrt{2}$}
	\label{fig:dvt}
\end{figure*}

\begin{figure*}[thb]
    \centering
    \vspace{-1.2in}
   \includegraphics[width=.45\textwidth]{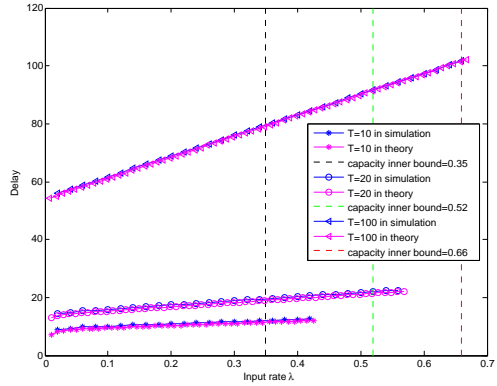}
  \includegraphics[width=.45\textwidth]{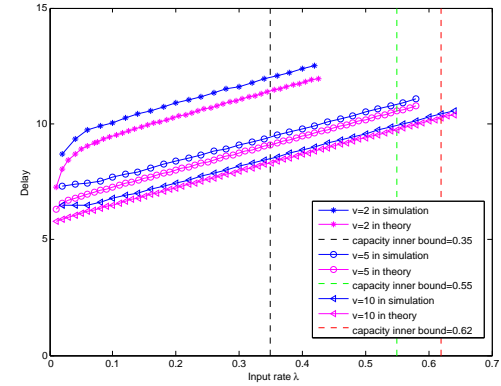}
    \vspace{-27.8mm}
 \caption{Delay of one-flow two-robots as we vary $T$ for $v=2$ (left) and we vary $v$ for $T=10$ (right)}
	\label{fig:1d}
\end{figure*}

\section{Simulations}
We first present numerical simulation results for two flows and four robots. We use $\eta = 2$, $C = 1$, $d_1 = 25, d_2 = 100$, $v$ and $T$ and $\lambda_i$ are varied as shown in the figures. In figure~\ref{fig:dvv} we see the average end-to-end delay (time taken for a packet generated at the source to reach the sink; it is obtained by measuring the average total queue size for each flow in the simulations and dividing by the arrival rate, as per Little's Theorem~\cite{garcia93}) versus arrival rate for each flow, plotted wherever CBMF results in stable queues; we find that we are able to get converging, bounded delays (indicative of stability) even beyond the inner-bound capacity line. Also marked on the figure is the lower (inner) bound of capacity, for rates below which CBMF is provably stable. We see that as the velocity increases, so does the capacity, and at the same time the delay decreases. Thus improvement in robot velocity benefits both throughput and delay performance of CBMF, as may be expected. Figure~\ref{fig:dvt}, in which the velocity is kept constant across curves but the epoch duration is varied, is somewhat similar but with one striking difference, however, as the epoch duration increases, so does the capacity; but at the same time, the average delay also increases (for the same arrival rates, so long stability is maintained). Thus, increasing the scheduling epoch duration improves throughput but hurts delay performance.

Second, we present numerical simulation results for delay in one-flow two-robots system. We use $d=10$, $T$, $v$ and $\lambda$ are varied as shown in Figure~\ref{fig:1d}. As it is shown in Figure~\ref{fig:1d}, the simulation results match the theoretical results. We can reach the same conclusion as before: increasing the scheduling epoch duration improves throughput but hurts delay performance; however, increasing robot velocities benefits both throughput and delay performance. Besides, in the Figure 5, delay seems to be a linear function of the input rate $\lambda$. This can be shown theoretically if we ignore the integral part of equation (27) or equation (32).






\section{Conclusions}

This paper has addressed two fundamental questions in robotic message ferrying for wireless networks: what is the throughput capacity region of such systems? How can they be scheduled to ensure stable operation, even without prior knowledge of arrival rates? There are a number of open directions suggested by the present work. The first is to improve the CBFM algorithm to support the entire capacity region without delay inefficiency, possibly by adapting the schedule length based on observed delay or by considering finer-grained motion control. Finally, we are interested in developing decentralized scheduling mechanisms that the robots can implement in a distributed fashion.

\end{document}